\documentclass[11pt]{article}

\usepackage{acl}

\usepackage{times}
\usepackage{latexsym}
\usepackage[T1]{fontenc}
\usepackage[utf8]{inputenc}
\usepackage{microtype}
\usepackage{graphicx}
\usepackage{amsmath,amssymb,amsfonts}
\usepackage{amsthm}
\newtheorem{proposition}{Proposition}
\usepackage{booktabs}
\usepackage{multirow}
\usepackage{listings}
\usepackage{footnote}

\title{Flip, Don't Shuffle: Watermarking LLMs at the Speed of Inference}

\author{
  Simone Ceppi \\
  European Commission \\
  Joint Research Centre \\
  Ispra, Italy \\
  \texttt{simone.ceppi@ec.europa.eu}
  \And
  Ignacio Sanchez \\
  European Commission \\
  Joint Research Centre \\
  Ispra, Italy \\
  \texttt{ignacio.sanchez@ec.europa.eu}
}

\begin{document}
\maketitle
\begin{abstract}
We introduce Stateless Bernoulli Watermarking (SBW), a new statistical watermark for Large Language Models that determines green list membership through independent per-token Bernoulli trials. Unlike KGW's vocabulary permutation or SynthID's multi-layer tournament, SBW requires only a single comparison per token against a counter-based random number generator, reducing membership complexity to $O(1)$ and enabling single-kernel execution with zero intermediate allocations. We prove that this formulation preserves the same detection guarantees as fixed-size green lists: the z-score test remains $\mathcal{N}(0,1)$ under the null. The stateless architecture enables capabilities unavailable to existing methods: full-vocabulary self-salt watermarking (over 6000$\times$ faster than KGW's self-salt and 2$\times$ faster than SynthID despite biasing the entire vocabulary with candidate-dependent seeding) and architectural compatibility with distributed inference. In end-to-end generation benchmarks, SBW adds less than 1\% overhead at all batch sizes. We additionally identify hash function design as a previously unexplored axis for watermark quality, showing that a GPU-native Jenkins hash improves null calibration by 1.8$\times$ while producing more diverse text. Experiments across two seeding schemes and eight $(\gamma, \delta)$ configurations confirm statistical equivalence with ROC-AUC differences below 0.01.
\end{abstract}

\section{Introduction and State of the Art}

As Large Language Models (LLMs) become integrated into the global information infrastructure, the ability to distinguish machine-generated text from human-written content has become a critical safety requirement. Emerging regulations across multiple jurisdictions now have started to require providers of AI systems to mark their output as machine-generated, making it urgent to research watermarking solutions deployable in production without incurring the so-called ``watermark tax'' (i.e.\ the significant computational and temporal overhead related to applying the watermark during inference).

Different methods exist to apply a watermark during text generation~\cite{liu2024taxonomy,zhang2024watermarking}. These include adversarial learning approaches that train a model to embed a retrievable signal~\cite{abdelnabi2021awt,li2023plmmark}, with recent work showing that decoding-based watermarks can also be distilled into model weights~\cite{gu2024learnability}, post-hoc methods that modify text after generation (e.g., synonym substitution~\cite{munyer2023natural}, semantic word insertion~\cite{chang2024postmark}), and statistical methods that modify the token sampling process itself~\cite{kirchenbauer2023watermark,christ2023undetectable,kuditipudi2023robust,hu2023unbiased,wu2024dipmark,dathathri2024scalable,zhao2025permute}. This work focuses on the last category, introduced by \citet{kirchenbauer2023watermark}, where green list biasing preserves output quality while showing resilience to basic paraphrasing attacks~\cite{kirchenbauer2024reliability}, though adaptive attackers can degrade detection rates~\cite{diaa2025adaptive,rastogi2024revisiting} and even steal watermark schemes from API access~\cite{jovanovic2024stealing}. Google's SynthID-Text~\cite{dathathri2024scalable} demonstrated production-scale deployment, confirming industry demand for efficient watermarking. However, KGW involves an inherently unfusible operation (vocabulary permutation via \texttt{randperm}), while SynthID requires $m=30$ sequential reweighting passes that impose per-token compute proportional to tournament depth. As vocabulary sizes grow beyond 100K tokens and deployments scale to millions of concurrent users, a watermarking method with $O(1)$ per-token cost compiling into a single fused kernel with no intermediate allocations would eliminate this overhead entirely.

\newpage

\textbf{Contributions.} We introduce \emph{Stateless Bernoulli Watermarking} (SBW), which replaces KGW's vocabulary permutation and SynthID's tournament passes with independent per-token Bernoulli trials: a local, stateless decision. We prove this preserves identical detection guarantees under the standard null model (Section~\ref{sec:theory}); experiments confirm this equivalence holds empirically, where SBW and KGW deviate identically from the idealized $\mathcal{N}(0,1)$ on real LLM generations (Section~\ref{sec:experiments}). The stateless formulation enables: (1) single-kernel $O(1)$ watermarking with zero intermediate allocation, (2) true full-vocabulary self-salt without top-$k$ approximation, and (3) compatibility with distributed inference (Section~\ref{sec:capabilities}). In benchmarks, SBW achieves 2$\times$ lower latency than SynthID and over 6000$\times$ lower than KGW, while being the first method to support full-vocabulary self-salt at production scale (Section~\ref{sec:scaling}). We provide a production-ready implementation integrated with vLLM~\cite{kwon2023efficient}. We additionally identify hash function design as a previously unexplored factor in watermark quality (Section~\ref{sec:hash}).

\textbf{Code.} Our implementation, including the vLLM logits processor, is available at \url{https://github.com/si-mon-jinn/sbw} (\texttt{pip install sbw}); watermark evaluation library used to run experiments at \url{https://github.com/si-mon-jinn/waterpipe}; paper source at \url{https://github.com/si-mon-jinn/flip-dont-shuffle}.

\section{Stateless Bernoulli Watermarking}
\label{sec:method}

\subsection{From Global Operations to Local Decisions}

Existing statistical watermarks involve costly additional computation at each generation step: KGW calls \texttt{randperm(V)} to partition the vocabulary, while SynthID applies $m$ sequential reweighting passes over sampled candidates. These approaches introduce overhead through: (1) per-token complexity that scales with vocabulary size ($O(V \log V)$ for KGW) or tournament depth ($O(k \cdot m)$ for SynthID), and (2) allocation of intermediate tensors (a full $(B, V)$ permutation for KGW; $(B, k, m)$ g-values for SynthID) that compete with KV-cache for GPU memory. Our method eliminates both by reconceptualizing green list membership as a local, stateless computation. 

\subsection{Bernoulli Green List Construction}

We propose \textbf{independent Bernoulli selection} for green list construction. For each token step $t$, we derive a seed $r_t$ from the preceding context via a pseudorandom function (PRF). Instead of shuffling the entire vocabulary, we treat each token $v \in V$ as an independent candidate for the green list $G_t$ via a threshold test:
\begin{equation}
G_t = \{ v \in V \mid \text{CBRNG}(v, r_t) < \gamma \}.
\end{equation}

The modified logits $l'_t$ are computed as in the KGW method:
\begin{equation}
l'_{t,v} = \begin{cases} l_{t,v} + \delta & \text{if } v \in G_t \\ l_{t,v} & \text{if } v \notin G_t \end{cases}.
\end{equation}

The key insight is that a counter-based random number generator (CBRNG), given a seed and a position, computes the random value at that position in $O(1)$ without computing any preceding values. This allows us to determine the ``greenness'' of any token ID independently, turning a global coordination problem into a local, parallelizable computation. Our implementation uses Philox 4x32-10~\cite{salmon2011parallel}, a CBRNG with strong statistical properties and native GPU support.

The practical benefits are substantial:
\begin{itemize}
\item \textbf{Generation:} Per-token complexity drops from $O(V \log V)$ (KGW) or $O(k \cdot m)$ (SynthID) to $O(V)$ with a single fused kernel, enabling significant speedups (Section~\ref{sec:fusion}).
\item \textbf{Self-salt schemes:} Full-vocabulary evaluation becomes tractable. KGW's self-salt requires a separate \texttt{randperm} per candidate, and SynthID does not support self-salt (Section~\ref{sec:selfsalt}).
\item \textbf{Detection:} Each token can be tested without computing the full green list, reducing per-token detection complexity from $O(V \log V)$ (KGW) or $O(m)$ (SynthID) to $O(1)$.
\end{itemize}

The theoretical cost is that $|G_t|$ becomes a random variable with $\mathbb{E}[|G_t|] = \gamma |V|$ rather than exactly $\gamma |V|$. Section~\ref{sec:theory} formally proves this has no impact on detection power, and Section~\ref{sec:zscore} confirms empirically that z-score distributions are indistinguishable.

\subsection{Jenkins Integer Hash}

For seeding schemes that hash token IDs (\texttt{selfhash}, \texttt{minhash}, \texttt{skipgram}), KGW uses a pre-generated permutation table to map tokens to ${\sim}10^6$ distinct values. This requires a memory lookup per hash call.

We replace this with the Bob Jenkins integer hash~\cite{jenkins1997hash}, that compiles into pure ALU instructions with full throughput avoiding scattered memory accesses that thrash GPU caches. This substitution does not affect watermark security, which rests on the secret key. The larger output range also has consequences for watermark quality, which we analyze in Section~\ref{sec:hash}.

\section{Theoretical Analysis of Equivalence}
\label{sec:theory}

We prove that replacing the fixed-size green list of KGW with stochastic Bernoulli membership preserves the detection z-score exactly.

\begin{proposition}
Let each token $i \in V$ be included in the green list $G_t$ independently with $X_i \sim \mathrm{Bernoulli}(\gamma)$, so the green list proportion
\begin{equation}
\Gamma_t = \frac{1}{|V|}\sum_{i=1}^{|V|} X_i
\end{equation}
is stochastic with $\mathrm{E}[\Gamma_t] = \gamma$. Under $H_0$ (text generated without watermarking, so the model is unaware of $G_t$), the per-token green indicator $Y_t$ satisfies $\mathrm{E}[Y_t] = \gamma$ and $\mathrm{Var}(Y_t) = \gamma(1-\gamma)$, identical to the fixed-size case.
\end{proposition}

\begin{proof}[Proof sketch]
By symmetry of the i.i.d.\ construction, $P(Y_t = 1 \mid \Gamma_t) = \Gamma_t$ (Appendix~\ref{app:derivation}). By the Law of Total Expectation, $\mathrm{E}[Y_t] = \mathrm{E}[\Gamma_t] = \gamma$. Applying the Law of Total Variance:
\begin{align}
\mathrm{Var}(Y_t)
  &= \underbrace{E[\Gamma_t(1{-}\Gamma_t)]}_{\text{cond.\ variance}}
  + \underbrace{\mathrm{Var}(\Gamma_t)}_{\text{variance of mean}} \nonumber\\
  &= E[\Gamma_t] - E[\Gamma_t^2] + \mathrm{Var}(\Gamma_t).
\end{align}
Substituting $E[\Gamma_t^2] = \mathrm{Var}(\Gamma_t) + \gamma^2$, the $\mathrm{Var}(\Gamma_t)$ terms cancel, yielding $\mathrm{Var}(Y_t) = \gamma - \gamma^2 = \gamma(1-\gamma)$.
\end{proof}

Since the $Y_t$'s are mutually independent under $H_0$ (Appendix~\ref{app:derivation}), the green token count $W = \sum_{t=1}^T Y_t$ is binomial with parameters $(T, \gamma)$, and the standard KGW z-score applies unchanged:
\begin{equation}
z = \frac{W - T\gamma}{\sqrt{T\gamma(1-\gamma)}} \sim \mathcal{N}(0,1).
\end{equation}

At detection, the stochastic formulation preserves the null distribution exactly, inheriting the asymptotic optimality of the z-score test for green-list watermarks~\cite{li2025statistical}, while enabling an $O(1)$ per-token membership test that makes kernel fusion possible (Section~\ref{sec:fusion}). The full derivation, including concentration bounds and independence proof, is in Appendix~\ref{app:derivation}.

\section{Numerical Experiments}
\label{sec:experiments}

To validate our theoretical equivalence claims, we conducted systematic experiments comparing KGW against Stateless Bernoulli Watermarking across z-score distributions (Section~\ref{sec:zscore}), detection accuracy (Section~\ref{sec:roc}), and text quality (Section~\ref{sec:perplexity}).

\subsection{Experimental Setup}

\textbf{Models and data.} We used Qwen3-8B~\cite{yang2025qwen3} (vocabulary size 151,936) for generation and Qwen3.5-27B~\cite{qwen2025qwen35omni} as an external perplexity judge. Prompts (500) were sampled from the C4 validation split~\cite{raffel2020exploring}, each truncated to 30 tokens. All experiments ran on a single NVIDIA RTX 3090 using temperature $T=1.0$ (pure sampling without top-p truncation). Full reproducibility details are in Appendix~\ref{app:reproducibility}. Validation on Falcon-7B with different sampling ($T{=}0.7$, $top\_p{=}0.8$) and hardware (A6000) confirms the results generalize beyond this setup (Appendix~\ref{app:generalization}).

\textbf{Watermark schemes.} We focus on two seeding strategies from \citet{kirchenbauer2023watermark}: context-only seeding (\texttt{simple\_1}, context width 1) and self-salt seeding (\texttt{selfhash}, context width 4 plus candidate token). Table~\ref{tab:schemes} summarizes all schemes evaluated. Statistical equivalence (this section) is validated on SBW-1 and SBW-ss, which match the KGW baselines in context width and candidate coverage. Performance benchmarks (Section~\ref{sec:scaling}) additionally include full-vocabulary variants (SBW-4, SBW-ss-V) that bias the entire 151K vocabulary.

\begin{table}[!t]
\caption{Watermark Schemes Evaluated. Top: statistical equivalence experiments. Bottom: performance benchmarks only.}
\label{tab:schemes}
\begin{center}
\footnotesize
\resizebox{\columnwidth}{!}{
\begin{tabular}{@{}lccccc@{}}
\toprule
Scheme & Green list & Context & Self-salt & Candidates & Section \\
\midrule
\texttt{simple\_1} & randperm & width 1 & No & full $V$ & \ref{sec:experiments} \\
SBW-1 & Bernoulli & width 1 & No & full $V$ & \ref{sec:experiments} \\
\texttt{selfhash} & randperm & width 4 & Yes & top-40 & \ref{sec:experiments} \\
SBW-ss & Bernoulli & width 4 & Yes & top-40 & \ref{sec:experiments} \\
\midrule
SBW-4 & Bernoulli & width 4 & No & full $V$ & \ref{sec:scaling} \\
SBW-ss-V & Bernoulli & width 4 & Yes & full $V$ & \ref{sec:scaling} \\
\bottomrule
\end{tabular}
}
\end{center}
\vspace{-2mm}
\end{table}

All SBW variants use fused compiled kernels; SBW-ss additionally uses Jenkins hash instead of KGW's permutation table.

\textbf{Parameter grid.} We evaluated each scheme across a grid of watermark strength parameters:
\begin{itemize}
\item $\gamma \in \{0.25, 0.50\}$: green list fraction (smaller = more constrained vocabulary)
\item $\delta \in \{1, 2, 5, 10\}$: logit bias strength (larger = stronger watermark signal)
\end{itemize}

\subsection{Z-Score Distribution Analysis}
\label{sec:zscore}

We verify that our Bernoulli method produces indistinguishable z-score distributions using Kolmogorov-Smirnov (KS) tests and moment comparisons.

\subsubsection{Without Self-Salt (\texttt{simple\_1} vs SBW-1)}

\textbf{Non-watermarked text.} We pool across all $\delta$ values (thus a sample with $N=4000$ completions). Under $H_0$ (no watermark applied, model unaware of $G_t$), the z-score of non-watermarked text should follow $\mathcal{N}(0,1)$. We first test each scheme's conformance using one-sample KS tests (Table~\ref{tab:ks_onesample}, top panel; Fig.~\ref{fig:simple1}, top right). Both schemes reject $\mathcal{N}(0,1)$ for both $\gamma$; this is expected~\cite{kirchenbauer2023watermark}: real LLM-generated text exhibits correlations between token distributions and hash-based green list assignments. The KS statistics are comparable between schemes (0.050--0.068), and the empirical means and standard deviations closely match, confirming our theoretical prediction. More directly, two-sample KS tests between the schemes show negligible effect sizes (Appendix~\ref{app:ks}, Table~\ref{tab:ks_twosample_nowm}): 0.026 (p=0.14) at $\gamma=0.25$ and 0.042 (p=0.002) at $\gamma=0.50$. While the p-value at $\gamma=0.50$ falls below 0.05 (expected given $N=4000$, which gives the test power to detect trivial differences), the maximum CDF difference is under 5 percentage points, confirming equivalent null distributions.

\begin{table}[!t]
\caption{One-Sample KS Tests vs $\mathcal{N}(0,1)$ (Non-Watermarked)}
\label{tab:ks_onesample}
\begin{center}
\footnotesize
\resizebox{\columnwidth}{!}{\begin{tabular}{@{}cccccc@{}}
\toprule
$\gamma$ & Scheme & Mean & Std & KS & p-value \\
\midrule
\multicolumn{6}{c}{\textit{Without self-salt}} \\
\midrule
0.25 & \texttt{simple\_1} & +0.120 & 0.991 & 0.062 & 8.65e-14 \\
0.25 & \texttt{SBW-1} & +0.091 & 0.965 & 0.057 & 7.29e-12 \\
0.50 & \texttt{simple\_1} & +0.053 & 1.036 & 0.053 & 2.16e-10 \\
0.50 & \texttt{SBW-1} & +0.138 & 0.995 & 0.068 & 1.05e-16 \\
\midrule
\multicolumn{6}{c}{\textit{With self-salt}} \\
\midrule
0.25 & \texttt{selfhash} & $-$0.518 & 1.155 & 0.239 & $< 10^{-267}$ \\
0.25 & \texttt{SBW-ss} & $-$0.272 & 1.159 & 0.137 & $< 10^{-66}$ \\
0.50 & \texttt{selfhash} & $-$0.205 & 1.262 & 0.139 & $< 10^{-68}$ \\
0.50 & \texttt{SBW-ss} & $-$0.096 & 1.140 & 0.087 & 7.02e-27 \\
\bottomrule
\end{tabular}}
\end{center}
\vspace{-2mm}
{\small $N=4000$ pooled across $\delta$ and $\gamma$ values.}
\end{table}

\textbf{Watermarked text.} We now verify that the watermark signal (the z-score shift induced by green list biasing) is preserved by our method. Table~\ref{tab:zscore_wm} (top panel) reports summary statistics and two-sample t-tests comparing the mean z-scores. The results confirm that both methods produce statistically equivalent watermark strength: 7 of 8 t-tests fail to reject the null hypothesis of equal means. The single significant difference is small in magnitude relative to the mean (Table~\ref{tab:zscore_wm}). The standard deviations also match closely, further validating our theoretical moment equivalence.

\begin{table}[!t]
\caption{Watermarked Z-Score Statistics}
\label{tab:zscore_wm}
\begin{center}
\footnotesize
\resizebox{\columnwidth}{!}{\begin{tabular}{@{}ccccccc@{}}
\toprule
$\gamma$ & $\delta$ & KGW (std) & Ours (std) & $\Delta$ & t-stat & p \\
\midrule
\multicolumn{7}{c}{\textit{Without self-salt}} \\
\midrule
0.25 & 1.0 & 3.21 (1.44) & 3.18 (1.33) & $-$0.03 & 0.40 & 0.693 \\
0.25 & 2.0 & 6.93 (2.07) & 6.69 (2.10) & $-$0.24 & 1.81 & 0.070 \\
0.25 & 5.0 & 15.54 (3.43) & 15.14 (3.53) & $-$0.40 & 1.83 & 0.068 \\
0.25 & 10.0 & 17.63 (4.66) & 17.64 (4.50) & +0.01 & -0.03 & 0.978 \\
0.50 & 1.0 & 3.19 (1.26) & 3.22 (1.22) & +0.03 & -0.41 & 0.679 \\
0.50 & 2.0 & 5.80 (1.63) & 5.92 (1.58) & +0.13 & -1.24 & 0.215 \\
0.50 & 5.0 & 10.26 (1.93) & 10.61 (1.57) & +0.35 & -3.16 & 0.002 \\
0.50 & 10.0 & 11.77 (1.98) & 11.92 (1.88) & +0.15 & -1.22 & 0.224 \\
\midrule
\multicolumn{7}{c}{\textit{With self-salt}} \\
\midrule
0.25 & 1.0 & 2.68 (1.62) & 2.96 (1.45) & +0.28 & -2.87 & 0.004 \\
0.25 & 2.0 & 6.67 (2.22) & 6.86 (2.16) & +0.19 & -1.39 & 0.165 \\
0.25 & 5.0 & 16.37 (3.61) & 16.66 (3.66) & +0.29 & -1.25 & 0.211 \\
0.25 & 10.0 & 18.31 (5.13) & 19.75 (4.85) & +1.44 & -4.55 & $< 10^{-5}$ \\
0.50 & 1.0 & 3.07 (1.36) & 3.11 (1.36) & +0.04 & -0.52 & 0.602 \\
0.50 & 2.0 & 6.10 (1.72) & 6.17 (1.63) & +0.07 & -0.67 & 0.504 \\
0.50 & 5.0 & 10.74 (2.05) & 11.17 (1.66) & +0.43 & -3.65 & $< 10^{-3}$ \\
0.50 & 10.0 & 12.64 (2.03) & 12.73 (2.01) & +0.09 & -0.67 & 0.505 \\
\bottomrule
\end{tabular}
}
\end{center}
\vspace{-2mm}
{\small $N=500$ per group.}
\end{table}

The top left panel of Fig.~\ref{fig:simple1} confirms this visually: the watermarked z-score distributions for both schemes overlap almost completely, with clear separation from zero.

\begin{figure}[!t]
\centerline{\includegraphics[width=\columnwidth]{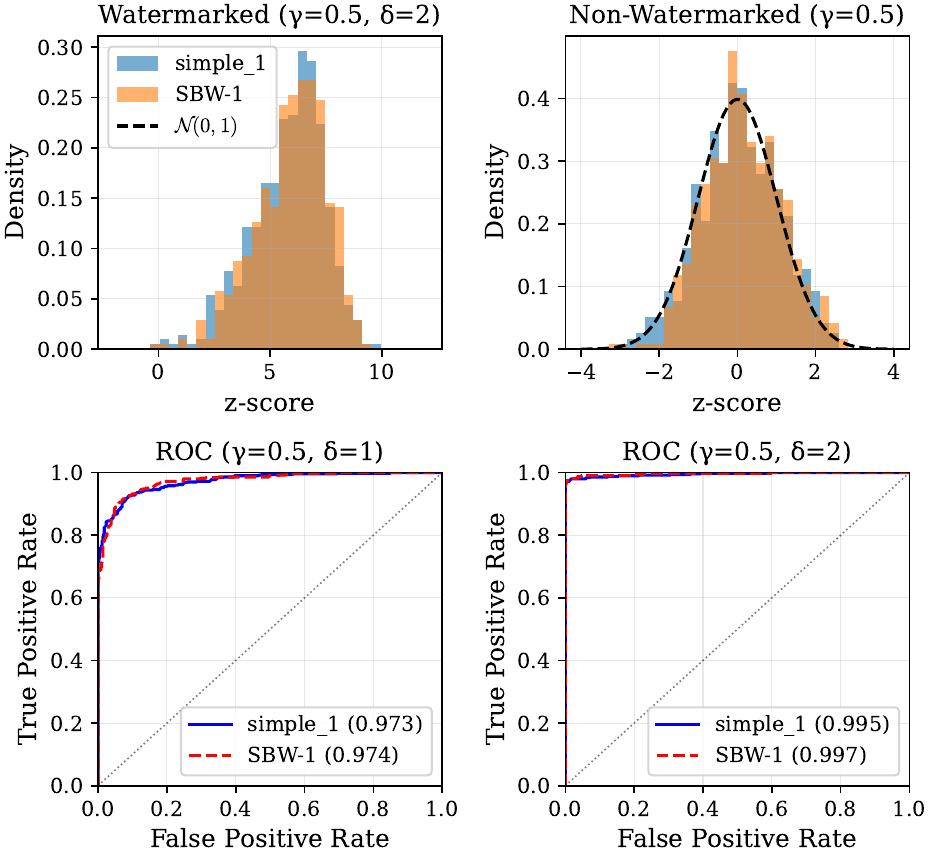}}
\caption{Without self-salt: z-score distributions are indistinguishable between implementations. Bottom: ROC curves showing equivalent detection (AUC differences $<$0.01).}
\label{fig:simple1}
\end{figure}

To test full distributional agreement beyond the mean, we apply two-sample KS tests to the watermarked z-scores for each configuration (Appendix~\ref{app:ks}, Table~\ref{tab:ks_twosample_wm}). Seven of eight configurations fail to reject the null hypothesis of identical distributions (p $>$ 0.05). The single rejection occurs at the same configuration flagged by the t-test, consistent with chance given 8 independent tests at $\alpha=0.05$.

\subsubsection{With Self-Salt (\texttt{selfhash} vs SBW-ss)}

\textbf{Non-watermarked text.} Unlike the non-self-salt case, the two-sample KS test strongly rejects the hypothesis that \texttt{selfhash} and SBW-ss produce the same null distribution ($D=0.112$, $p < 10^{-22}$ at $\gamma=0.25$; Appendix~\ref{app:ks}, Table~\ref{tab:ks_twosample_nowm}). However, SBW-ss is consistently \emph{closer} to the theoretical $\mathcal{N}(0,1)$: 1.74$\times$ at $\gamma=0.25$ and 1.60$\times$ at $\gamma=0.50$ (measured by KS statistic ratio; Table~\ref{tab:ks_onesample}, bottom panel; Fig.~\ref{fig:selfhash}, top right), with the negative mean bias 1.9$\times$ smaller at $\gamma=0.25$ and 2.1$\times$ at $\gamma=0.50$. Section~\ref{sec:hash} attributes this difference entirely to the hash function, not the Bernoulli construction.

\textbf{Watermarked text.} Table~\ref{tab:zscore_wm} bottom panel reports summary statistics and two-sample t-tests for watermarked text. The SBW-ss scheme consistently produces slightly higher z-scores (all $\Delta$ positive), with 3 of 8 configurations reaching statistical significance.

Fig.~\ref{fig:selfhash} top left shows the results for the selfhash scheme, confirming the slight positive shift for SBW-ss.

\begin{figure}[!t]
\centerline{\includegraphics[width=\columnwidth]{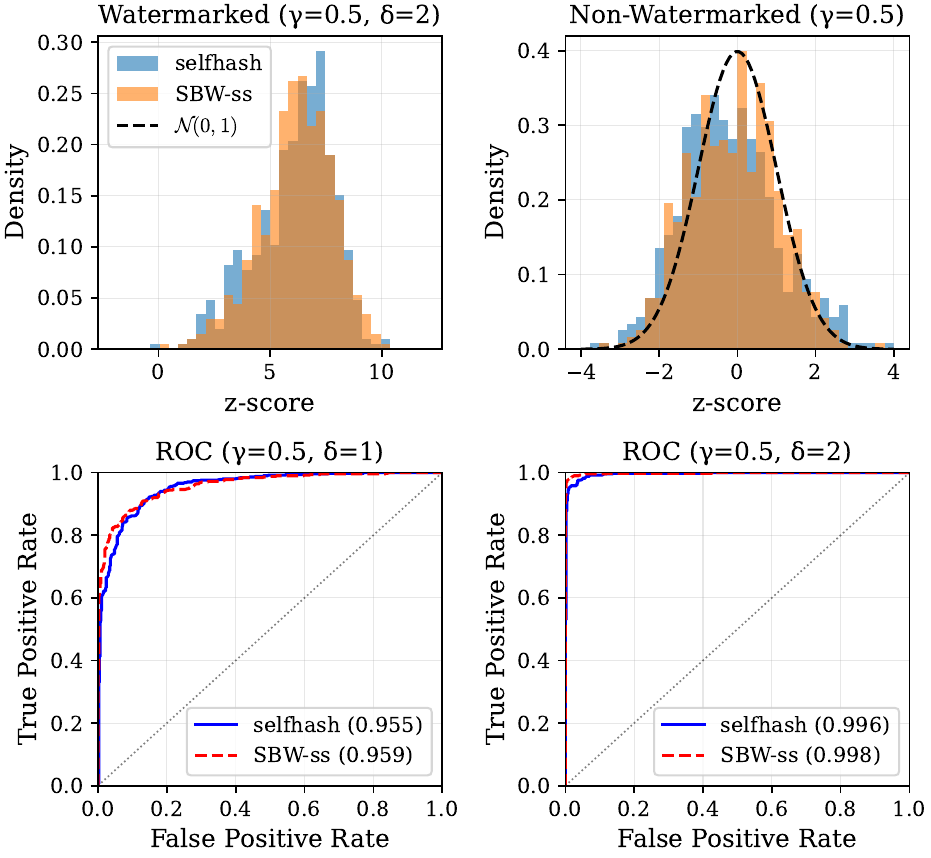}}
\caption{With self-salt: SBW-ss produces slightly higher z-scores due to the Jenkins hash (Section~\ref{sec:hash}). Bottom: ROC curves showing equivalent detection.}
\label{fig:selfhash}
\end{figure}

Two-sample KS tests confirm the distributional differences (Appendix~\ref{app:ks}, Table~\ref{tab:ks_twosample_wm}): four of eight configurations reject distributional equivalence, contrasting with the non-self-salt result (1/8 rejections). The effect is stronger at high $\delta$ values, where the watermark bias amplifies any difference in green list composition.

\subsection{Detection Accuracy: ROC-AUC Analysis}
\label{sec:roc}

We measure detection performance using ROC-AUC, which captures the trade-off between true positive rate and false positive rate across all z-score thresholds.

\subsubsection{Without Self-Salt (\texttt{simple\_1} vs SBW-1)}

Given the z-score equivalence established above, we expect equivalent detection performance. Table~\ref{tab:roc} confirms this prediction.

\begin{table}[!t]
\caption{ROC-AUC Comparison}
\label{tab:roc}
\begin{center}
\footnotesize
\resizebox{\columnwidth}{!}{\begin{tabular}{@{}ccccc|ccc@{}}
\toprule
\multicolumn{5}{c|}{\textit{Without self-salt}} & \multicolumn{3}{c}{\textit{With self-salt (top-40)}} \\
$\gamma$ & $\delta$ & KGW & Ours & $\Delta$ & KGW & Ours & $\Delta$ \\
\midrule
0.25 & 1.0 & .96 & .97 & +.009 & .943 & .953 & +.009 \\
0.25 & 2.0 & .995 & .997 & +.001 & .995 & .994 & $-$.001 \\
0.25 & 5.0 & 1.00 & 1.00 & .000 & 1.00 & 1.00 & .000 \\
0.25 & 10 & 1.00 & 1.00 & .000 & 1.00 & 1.00 & .000 \\
0.50 & 1.0 & .972 & .974 & +.001 & .955 & .959 & +.004 \\
0.50 & 2.0 & .995 & .996 & +.001 & .995 & .998 & +.002 \\
0.50 & 5.0 & 1.00 & 1.00 & .000 & 1.00 & 1.00 & .000 \\
0.50 & 10 & 1.00 & 1.00 & .000 & .999 & 1.00 & +.000 \\
\bottomrule
\end{tabular}
}
\end{center}
\vspace{-2mm}
{\small 500 WM + 500 non-WM per config. Max $\Delta$AUC = 0.01.}
\end{table}

The maximum AUC difference is below 1\% (Table~\ref{tab:roc}) and occurs at the weakest configuration where sampling variance is highest. Across all configurations, the differences show no systematic bias favoring either implementation. At higher $\delta$ values, both schemes achieve near-perfect separation (AUC $\geq$ 0.999), leaving no room for measurable differences. Fig.~\ref{fig:simple1} bottom row confirms this visually: the ROC curves are indistinguishable. These results demonstrate that Stateless Bernoulli Watermarking achieves equivalent detection accuracy to KGW.

\subsubsection{With Self-Salt (\texttt{selfhash} vs SBW-ss)}

Since SBW-ss produces slightly higher z-scores on watermarked text (Section~\ref{sec:zscore}), detection performance can only remain equivalent or improve. Fig.~\ref{fig:selfhash} bottom panel confirms this: AUC differences are below 1\% (Table~\ref{tab:roc}), with SBW-ss showing a marginal advantage at low false positive rates due to its better-calibrated null distribution.

\subsection{Text Quality: Perplexity Analysis}
\label{sec:perplexity}

We measure text quality using perplexity (PPL) from the external evaluation model.

\subsubsection{Without Self-Salt (\texttt{simple\_1} vs SBW-1)}

Table~\ref{tab:ppl} and Fig.~\ref{fig:ppl} left panels report the perplexity comparison for the non-self-salt scheme. The perplexity degradation follows expected patterns: larger $\delta$ values cause greater quality loss, and smaller $\gamma$ (more constrained green lists) amplifies this effect. Critically, both implementations produce nearly identical perplexity profiles. Only one of eight configurations shows p $<$ 0.05 (Table~\ref{tab:ppl}), consistent with chance given 8 independent tests at $\alpha=0.05$.

\begin{table}[!t]
\caption{Perplexity Degradation ($\Delta$PPL = WM $-$ NoWM)}
\label{tab:ppl}
\begin{center}
\footnotesize
\resizebox{\columnwidth}{!}{\begin{tabular}{@{}cc|ccc|ccc@{}}
\toprule
& & \multicolumn{3}{c|}{\textit{Without self-salt}} & \multicolumn{3}{c}{\textit{With self-salt}} \\
$\gamma$ & $\delta$ & $\Delta$KGW & $\Delta$Ours & p & $\Delta$KGW & $\Delta$Ours & p \\
\midrule
0.25 & 1.0 & +0.25 & +0.17 & 0.54 & +0.19 & +0.24 & 0.26 \\
0.25 & 2.0 & +1.34 & +1.05 & 0.06 & +1.29 & +1.27 & 0.70 \\
0.25 & 5.0 & +6.91 & +6.00 & 0.01 & +7.16 & +8.14 & .004 \\
0.25 & 10 & +11.2 & +10.4 & 0.23 & +10.6 & +14.9 & $<10^{-7}$ \\
0.50 & 1.0 & +0.31 & +0.25 & 0.64 & +0.26 & +0.15 & 0.97 \\
0.50 & 2.0 & +0.98 & +1.07 & 0.71 & +1.11 & +1.07 & 0.80 \\
0.50 & 5.0 & +4.01 & +4.19 & 0.25 & +4.15 & +4.78 & .001 \\
0.50 & 10 & +6.63 & +6.74 & 0.71 & +7.66 & +8.40 & 0.03 \\
\bottomrule
\end{tabular}
}
\end{center}
\vspace{-2mm}
{\small $N=500$. Baseline PPL without watermark $\approx$ 5.4.}
\end{table}

\begin{figure}[!t]
\centerline{\includegraphics[width=\columnwidth]{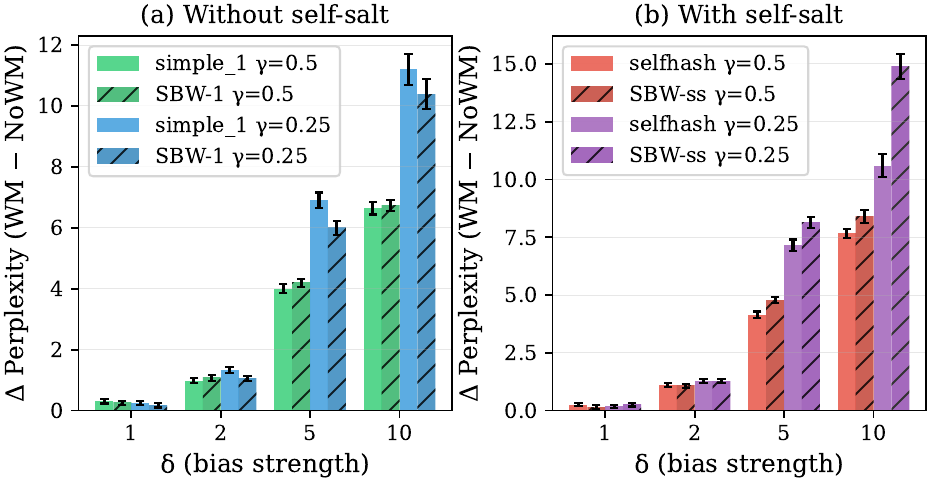}}
\caption{Perplexity degradation ($\Delta$PPL = watermarked $-$ non-watermarked) vs $\delta$. (a) Without self-salt: implementations are indistinguishable. (b) With self-salt: SBW-ss shows higher degradation at $\delta \geq 5$, attributable to the Jenkins hash (Section~\ref{sec:hash}).}
\label{fig:ppl}
\end{figure}

\subsubsection{With Self-Salt (\texttt{selfhash} vs SBW-ss)}

Unlike the non-self-salt comparison, Table~\ref{tab:ppl} right panel reveals a systematic difference: SBW-ss shows higher perplexity degradation at $\delta \geq 5$, with the gap growing substantially at aggressive configurations (Table~\ref{tab:ppl}). At practical settings ($\delta \leq 2$), the implementations are statistically indistinguishable. This pattern suggests the hash function interacts with the watermark bias; Section~\ref{sec:hash} attributes this to the Jenkins hash and shows the higher perplexity is accompanied by higher text diversity.

\subsection{Hash Function Design: A New Axis for Watermark Quality}
\label{sec:hash}

The experiments in Sections~\ref{sec:zscore}--\ref{sec:perplexity} compared \texttt{selfhash} and SBW-ss, which differ in two ways simultaneously: (1) green list construction (permutation vs Bernoulli) and (2) hash function (permutation table vs Jenkins). Because the hash difference produced unexpected results, we introduce a control scheme, SBW-ss-cpu, which uses the GPU Bernoulli code path but retains the KGW hash function. By comparing all three schemes on the same data, we isolate the hash effect and show that hash function design significantly impacts watermark quality, independent of the green list construction method.

\subsubsection{Attribution of Distributional Differences}
\label{sec:hash_attribution}

To confirm the distributional differences observed in Section~\ref{sec:zscore} and identify their cause, we score a separate set of 500 non-watermarked completions (from a $\delta=0$ control experiment) with three detectors: \texttt{selfhash}, SBW-ss-cpu (our control), and SBW-ss. The results (detailed in Appendix~\ref{app:hash}) are unambiguous: when the hash function is held constant, the Bernoulli method produces distributions indistinguishable from permutation ($p=0.173$ and $p=0.329$ for $\gamma=0.25, 0.5$ respectively), while the different-hash comparison is strongly rejected ($p=4{\times}10^{-4}$ at $\gamma=0.25$). The distributional differences are entirely attributable to the Jenkins hash function, not the Bernoulli green list construction.

On this control experiment, the Jenkins hash improves null calibration by 1.8$\times$ at $\gamma=0.25$ and reduces the negative mean bias by 2.3$\times$, in line with what already reported. At any fixed z-score threshold, the SBW-ss detector therefore produces fewer false positives. The effect is most pronounced at $\gamma=0.25$, where the smaller green list amplifies sensitivity to hash uniformity.

\subsubsection{Perplexity and Diversity}
\label{sec:hash_diversity}

Using the same control scheme, we confirm that the perplexity difference between \texttt{selfhash} and SBW-ss is also attributable to the hash function (Appendix~\ref{app:hash}). At practical configurations ($\delta \leq 2$), the two schemes produce statistically indistinguishable diversity and perplexity. At aggressive configurations ($\delta \geq 5$), a tradeoff emerges: \texttt{selfhash} produces lower perplexity but also lower diversity (e.g., diversity drops to 0.40 at $\delta=10$ vs 0.62 with Jenkins, relative to a 0.70 non-watermarked baseline), while SBW-ss produces higher perplexity but higher diversity. The lower perplexity in \texttt{selfhash} does not indicate better text quality; it reflects the model being constrained to repeat similar token patterns due to autocorrelation in the permutation table hash's green list assignments.

\section{Capabilities Enabled by Stateless Formulation}
\label{sec:capabilities}

Beyond the asymptotic improvement, the shift to a stateless, threshold-based RNG unlocks several hardware-level optimizations that we have implemented and validated in a production-ready codebase integrated with vLLM~\cite{kwon2023efficient}.

\subsection{Fused Kernel Architecture}
\label{sec:fusion}

Kernel fusion, combining multiple operations into a single GPU kernel launch, is critical for latency-sensitive inference as it eliminates intermediate memory traffic and kernel launch overhead. KGW's permutation cannot be fused: the GPU must complete the shuffle, materialize a $(B, V)$ index tensor, derive a boolean mask, and apply the bias, each as a separate kernel. SynthID's $m$ reweighting passes can be fused by \texttt{torch.compile} into a single kernel (as we demonstrate in our benchmark), but the fused kernel still performs $O(k \cdot m)$ work per token step.

Our Bernoulli formulation eliminates this barrier through three properties: (1) in-place logit mutation with no intermediate tensor materialization (zero additional memory required; Section~\ref{sec:scaling}), (2) a single CBRNG evaluation per token (no sequential dependencies), and (3) fully GPU-resident computation with no CPU synchronization. The entire watermarking operation can be expressed as a single pointwise tensor expression and compiled into one fused CUDA kernel using \texttt{torch.compile}~\cite{ansel2024pytorch2,tillet2019triton} with \texttt{mode="max-autotune"}, absorbing watermarking into existing logit post-processing without a dedicated pipeline step.

\subsection{Full-Vocabulary Self-Salt Watermarking}
\label{sec:selfsalt}

Self-salt seeding schemes, where the candidate token itself influences the seed, exhibit superior robustness to editing attacks~\cite{kirchenbauer2023watermark}. However, existing methods either avoid self-salt entirely (SynthID uses a fixed context window of $H=4$ tokens, independent of the candidate) or restrict it to a small subset: KGW evaluates only the top-$k=40$ candidates, leaving the remaining vocabulary unbiased and weakening the watermark signal in high-entropy contexts.

Our approach reduces the per-candidate cost to $O(1)$, making full-vocabulary self-salt practical. Setting $k = V$ biases the \textbf{entire vocabulary} at $O(B \times V)$ total cost, cheaper than KGW's $O(B \times 40 \times V \log V)$ for just 40 candidates. For the first time, true self-salt watermarking without approximation is feasible at production scale, combining the robustness benefits of candidate-dependent seeding with full vocabulary coverage.

\subsection{Inherited Security and Robustness}
\label{sec:security}

A natural question is whether SBW inherits the security and robustness guarantees established for KGW~\cite{kirchenbauer2023watermark,kirchenbauer2024reliability,fernandez2023three,zhao2023provable}. We argue that it does, inheriting both the strengths and the fundamental limitations~\cite{zhang2024impossibility} of the underlying statistical framework. The existing theoretical analyses rely on a single statistical property: conditioned on the seeding context, each token is assigned to the green list with probability $\gamma$, independently of other tokens. Our Proposition~1 proves that SBW satisfies exactly this property. Any theorem whose proof depends only on this independence and the marginal $\gamma$ transfers directly, including detection power and false positive rates (confirmed empirically by our two-sample KS tests, Table~\ref{tab:ks_twosample_nowm}), robustness to partial edits (same signal degradation for identical seeding schemes), resistance to removal attacks (same $\delta$ bias profile), and security against spoofing (inverting the CBRNG is computationally equivalent to inverting KGW's PRF-seeded permutation). Appendix~\ref{app:robustness} validates this empirically: across 19 attack configurations (character/word edits, truncation, paraphrasing, MLM substitution), SBW matches or exceeds (due to improved hash function; Section~\ref{sec:hash}) KGW's robustness in all cases.

\section{Performance Evaluation}
\label{sec:scaling}

To quantify practical impact, we measure end-to-end watermark overhead on Qwen2-7B, generating 128 tokens per request across batch sizes 16--256, with 50 iterations and 10 warmup runs. Both methods run as logits processors within the same \texttt{transformers.generate()} loop.\footnote{We use Transformers rather than vLLM because no vLLM-compatible SynthID processor exists; porting our vLLM processor to Transformers was straightforward, ensuring a fair comparison. Qwen2-7B was chosen as it is supported by the older Transformers version (4.43.3) required by the official SynthID implementation.} Overhead is computed as the difference in median generation time between watermarked and non-watermarked runs; error bars show $\pm 1\sigma$ computed as $\sqrt{\sigma_{\text{wm}}^2 + \sigma_{\text{base}}^2}$.

Fig.~\ref{fig:e2e_overhead} reports the results. SBW adds less than 1\% overhead across all batch sizes (9--76~ms over 3.4--14.5~s generation time), with the absolute overhead remaining nearly constant as batch size grows. SynthID (official \texttt{SynthIDSparseTopKMixin}, top-$k$=40) adds 80--148~ms (0.6--4.4\% at small batches). At batch 64, SBW overhead is 21~ms vs.\ SynthID's 136~ms, a \textbf{6.5$\times$ reduction}. At batch 256, the two converge: SBW adds 76~ms (0.53\%) while SynthID adds 80~ms (0.55\%).

The relative overhead (right panel) reveals distinct scaling regimes. SBW's percentage overhead remains approximately constant (${\sim}$0.5\%) across all batch sizes, while SynthID's decreases from 6\% to 0.6\%. This is explained by GPU compute saturation: SBW computes seeds for every token in the 151K vocabulary (full-vocab self-salt), performing ${\sim}$3800$\times$ more work per position than SynthID's top-40 approximation. This full-vocabulary computation saturates GPU compute around $B$=32--64, as confirmed by isolated kernel benchmarks (Appendix~\ref{app:performance}). Beyond this saturation point, the watermark kernel time grows proportionally with batch size, matching the growth rate of generation time, yielding a constant percentage overhead. SynthID's top-40 approximation requires far less compute and does not saturate, so its fixed absolute cost (${\sim}$100--200~ms) becomes a shrinking fraction of the growing generation time. The two curves cross at $B{\approx}$256; on more powerful hardware, this crossing point would shift to higher batch sizes (as the saturation threshold increases), but the crossing is inherent to the asymmetry in computational work.

These measurements are conservative for SBW: in production vLLM deployments, the watermark kernel runs on the same CUDA stream as the model forward pass, allowing partial overlap with memory-bound operations. Neither KGW (CPU-side \texttt{randperm}) nor SynthID (stateful mixin design) can benefit from this pipelining. A dedicated vLLM serving benchmark confirms ${\sim}$1\% overhead at production concurrency (Appendix~\ref{app:vllm}), with room for further optimization.

In isolated benchmarks measuring only the watermarking logits processor (excluding the LLM forward pass), SBW achieves \textbf{2--3$\times$ lower latency than SynthID}~\cite{dathathri2024scalable} and \textbf{over 6000$\times$ lower than KGW}~\cite{kirchenbauer2023watermark}, despite processing the entire 151K vocabulary with candidate-dependent seeding while SynthID processes only top-40 and KGW is limited to top-40 self-salt. Our schemes also use zero additional memory (in-place logit mutation), while KGW allocates $9 B \cdot V$ bytes per step and SynthID allocates tensors proportional to $B \times k \times \text{depth}$ (reaching 3.8~MB at batch 512). We note that SynthID provably preserves the output distribution (non-distortionary), while our method adds $\delta$ to green token logits with negligible perplexity impact at practical values ($\delta \leq 2$, Section~\ref{sec:perplexity}). Full isolated latency results, tail latency percentiles, speedup ratios, and memory overhead are reported in Appendix~\ref{app:performance}.

\begin{figure}[!t]
\centerline{\includegraphics[width=\columnwidth]{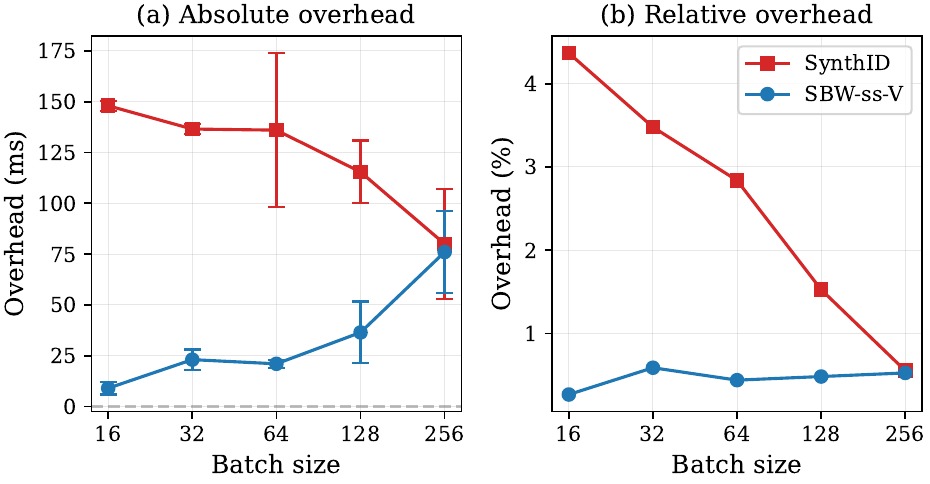}}
\caption{End-to-end watermark overhead (Qwen2-7B, RTX 3090, 128 output tokens, 50 iterations). Left: absolute overhead in ms with $\pm 1\sigma$ error bars. Right: relative overhead as percentage of total generation time. SBW remains below 1\% at all batch sizes while SynthID (official, top-$k$=40) ranges from 0.6\% to 4.4\%. KGW (selfhash, top-40) is omitted as out of scale: +1.8s at $B$=1 (57\%), +7.0s at $B$=4 (221\%), +27.7s at $B$=16 (869\%).}
\label{fig:e2e_overhead}
\end{figure}

\section{Conclusion}

We have presented Stateless Bernoulli Watermarking (SBW), a statistical watermark that achieves the same detection guarantees as KGW through a fundamentally different mechanism: independent per-token Bernoulli trials rather than inherently sequential vocabulary operations. This enables single-kernel $O(1)$ watermarking with zero allocation, full-vocabulary self-salt, and architectural compatibility with distributed inference, while achieving 2$\times$ lower latency than SynthID and over 6000$\times$ lower than KGW. A second contribution is the identification of hash function design as a meaningful factor in watermark quality. These results demonstrate that the ``watermark tax'' is not inherent to statistical watermarking but a consequence of inherently sequential formulations. The computational headroom freed by SBW opens future directions including vocabulary-agnostic watermarks that transfer across tokenizers, asymmetric schemes with public verification, and adaptive $\gamma$/$\delta$ strategies informed by token entropy~\cite{liu2024adaptive} or online optimization~\cite{cai2024statistical}.

\section{Ethical Considerations}

We advocate that watermarking should be applied transparently: users should be informed when a watermark is present. The method itself is agnostic to disclosure policy, and we encourage adopters to pair deployment with clear user-facing documentation.

Regarding false accusations, the z-score framework provides calibrated Type~I error control: at sufficiently conservative thresholds and with texts of adequate length, the false positive rate becomes negligible. We caution against applying detection to very short texts where statistical power is limited.

\section{Limitations}

The main experiments use a single GPU (RTX 3090) and model (Qwen3-8B); while we validate on an additional model and setup (Appendix~\ref{app:generalization}), broader evaluation across model families remains future work. While we validate robustness across 19 attack configurations (Appendix~\ref{app:robustness}), we did not evaluate adaptive watermark-removal strategies or spoofing attacks; dedicated security analysis remains future work. Our GPU implementation relies on compiler-generated Triton kernels via \texttt{torch.compile}; hand-tuned CUDA kernels with explicit memory coalescing and warp-level optimizations could further reduce latency, particularly at small batch sizes where kernel launch overhead dominates. Finally, while the stateless formulation is architecturally compatible with distributed and multi-device inference (each device can independently compute identical green lists from the shared key), we have not empirically validated this in a tensor-parallel or pipeline-parallel deployment.

\bibliography{custom}

\appendix

\section{Hash Function Analysis: Detailed Results}
\label{app:hash}

This appendix provides the full statistical analysis supporting the hash function findings in Section~\ref{sec:hash}.

\subsection{Distributional Attribution}

We scored 500 non-watermarked completions (from a separate $\delta=0$ control experiment) with three detectors: \texttt{selfhash}, SBW-ss-cpu (Bernoulli green list with permutation table hash), and SBW-ss (Bernoulli green list with Jenkins hash). For each pair of detectors (with $\gamma=0.25, 0.5$), we apply the two-sample Kolmogorov-Smirnov (KS) test. Table~\ref{tab:hash_attribution} reports the results.

\begin{table}[!t]
\caption{Two-Sample KS Tests Isolating Hash Effect}
\label{tab:hash_attribution}
\begin{center}
\footnotesize
\resizebox{\columnwidth}{!}{\begin{tabular}{@{}l|cc|cc@{}}
\toprule
& \multicolumn{2}{c|}{$\gamma=0.25$} & \multicolumn{2}{c}{$\gamma=0.50$} \\
Comparison & KS & p & KS & p \\
\midrule
simple\_1 vs SBW-1 (same hash) & 0.028 & 0.990 & 0.040 & 0.819 \\
selfhash vs SBW-ss-cpu (same) & 0.070 & 0.173 & 0.060 & 0.329 \\
selfhash vs SBW-ss (diff hash) & 0.130 & $4{\times}10^{-4}$ & 0.054 & 0.460 \\
\bottomrule
\end{tabular}
}
\end{center}
\vspace{-2mm}
{\small $N=500$. Only different-hash comparison rejected at $\gamma=0.25$.}
\end{table}

When the hash function is held constant, the Bernoulli method produces distributions indistinguishable from permutation. At $\gamma=0.50$, even the different-hash comparison fails to reject equivalence, indicating the effect is most pronounced when the smaller green list amplifies sensitivity to hash uniformity.

\subsection{Null Calibration}

We assess how well each detector's null distribution conforms to the theoretical $\mathcal{N}(0,1)$ using one-sample KS tests. A lower KS statistic indicates better calibration, meaning the detector's false positive rate at any z-score threshold is closer to the nominal rate. Table~\ref{tab:null_calibration} reports the results.

\begin{table}[!t]
\caption{One-Sample KS Tests vs $\mathcal{N}(0,1)$ (Null Calibration)}
\label{tab:null_calibration}
\begin{center}
\footnotesize
\resizebox{\columnwidth}{!}{\begin{tabular}{@{}l|cc|cc@{}}
\toprule
& \multicolumn{2}{c|}{$\gamma=0.25$} & \multicolumn{2}{c}{$\gamma=0.50$} \\
Detector & Mean & KS & Mean & KS \\
\midrule
selfhash (CPU hash) & $-$0.634 & 0.259 & $-$0.150 & 0.101 \\
SBW-ss-cpu & $-$0.694 & 0.296 & $-$0.104 & 0.067 \\
SBW-ss (Jenkins) & $-$0.279 & 0.143 & $-$0.102 & 0.091 \\
\bottomrule
\end{tabular}
}
\end{center}
\vspace{-2mm}
{\small $N=500$. Lower KS = closer to theoretical null.}
\end{table}

The Bernoulli method alone (with permutation table hash) does not improve calibration, as predicted by our equivalence proof (Section~\ref{sec:theory}), confirming that the improvement is attributable to the Jenkins hash's larger output range producing more uniform green list assignments.

\subsection{Perplexity Isolation}

Using the same SBW-ss-cpu control scheme, we compare perplexity across the three implementations. Table~\ref{tab:ppl_hash_isolation} shows that when using the same hash function, \texttt{selfhash} and SBW-ss-cpu produce statistically indistinguishable perplexity, confirming the hash function as the sole cause of the perplexity difference.

\begin{table}[!t]
\caption{Perplexity Comparison Isolating Hash Effect}
\label{tab:ppl_hash_isolation}
\begin{center}
\footnotesize
\resizebox{\columnwidth}{!}{\begin{tabular}{@{}ll|cc|cc@{}}
\toprule
& & \multicolumn{2}{c|}{$\delta=5$} & \multicolumn{2}{c}{$\delta=10$} \\
Scheme & Hash & $\Delta$ppl & p & $\Delta$ppl & p \\
\midrule
selfhash & CPU perm & 7.16 & --- & 10.60 & --- \\
SBW-ss-cpu & CPU perm & 7.81 & 0.095 & 10.67 & 0.927 \\
SBW-ss & Jenkins & 8.14 & .004 & 14.89 & .000 \\
\bottomrule
\end{tabular}
}
\end{center}
\vspace{-2mm}
{\small $\gamma=0.25$, $N=500$. p-values vs selfhash.}
\end{table}

\subsection{Diversity Analysis}

The permutation table hash, with its smaller range and table-based structure, produces more autocorrelation in green list assignments across similar contexts. This autocorrelation manifests as text repetition: when similar contexts map to similar seeds, the model is biased toward the same green tokens repeatedly. We quantify this using a diversity metric based on unique n-gram fractions~\cite{welleck2020neural,li2023contrastive,kirchenbauer2024reliability}:
\begin{equation}
\text{diversity} = -\log\left(1 - \prod_{n=1}^{N} u_n\right)
\end{equation}
where $u_n$ is the fraction of unique n-grams at order $n$ ($N=4$). Higher values indicate more diverse (less repetitive) text. Table~\ref{tab:diversity} reports the results.

\begin{table}[!t]
\caption{Text Diversity Comparison (Watermarked)}
\label{tab:diversity}
\begin{center}
\footnotesize
\begin{tabular}{@{}cccc@{}}
\toprule
$\delta$ & selfhash & SBW-ss & p \\
\midrule
1 & 0.696 & 0.723 & 0.13 \\
2 & 0.728 & 0.761 & 0.10 \\
5 & 0.695 & 0.805 & $<$.0001 \\
10 & 0.403 & 0.624 & $<$.0001 \\
\bottomrule
\end{tabular}

\end{center}
\vspace{-2mm}
{\small $\gamma=0.25$, $N=500$. NoWM diversity $\approx$ 0.70. Higher = less repetition.}
\end{table}

\section{Two-Sample KS Tests}
\label{app:ks}

\subsection{Non-Watermarked Z-Scores}

We pool non-watermarked z-scores across all $\delta$ values and generation $\gamma$ settings ($N=4000$ per scheme), then run two-sample KS tests between KGW and SBW for each detector $\gamma$. Table~\ref{tab:ks_twosample_nowm} reports the results.

\begin{table}[!t]
\caption{Two-Sample KS Tests on Non-Watermarked Z-Scores}
\label{tab:ks_twosample_nowm}
\begin{center}
\footnotesize
\begin{tabular}{@{}ccccl@{}}
\toprule
Scheme & $\gamma_{\text{det}}$ & $N$ & KS & p-value \\
\midrule
\texttt{simple\_1} & 0.25 & 4000 & 0.026 & 0.141 \\
\texttt{simple\_1} & 0.50 & 4000 & 0.042 & $< 10^{-3}$ \\
\midrule
\texttt{selfhash} & 0.25 & 4000 & 0.112 & $< 10^{-22}$ \\
\texttt{selfhash} & 0.50 & 4000 & 0.081 & $< 10^{-12}$ \\
\bottomrule
\end{tabular}

\end{center}
\vspace{-2mm}
{\small $N=4000$ pooled (4 $\delta$ $\times$ 2 $\gamma$ $\times$ 500 samples). Same texts scored by both detectors.}
\end{table}

\subsection{Watermarked Z-Scores}

To test full distributional agreement beyond the mean, we apply two-sample KS tests to the watermarked z-scores for each $(\gamma, \delta)$ configuration. Table~\ref{tab:ks_twosample_wm} reports the results for both schemes.

\begin{table}[!t]
\caption{Two-Sample KS Tests on Watermarked Z-Scores}
\label{tab:ks_twosample_wm}
\begin{center}
\footnotesize
\begin{tabular}{@{}cc|cc|cc@{}}
\toprule
& & \multicolumn{2}{c|}{\textit{Without self-salt}} & \multicolumn{2}{c}{\textit{With self-salt}} \\
$\gamma$ & $\delta$ & KS & p & KS & p \\
\midrule
0.25 & 1.0 & 0.044 & 0.719 & 0.098 & 0.016 \\
0.25 & 2.0 & 0.062 & 0.292 & 0.048 & 0.613 \\
0.25 & 5.0 & 0.080 & 0.082 & 0.096 & 0.020 \\
0.25 & 10 & 0.058 & 0.370 & 0.154 & 0.000 \\
0.50 & 1.0 & 0.034 & 0.935 & 0.050 & 0.560 \\
0.50 & 2.0 & 0.044 & 0.719 & 0.046 & 0.666 \\
0.50 & 5.0 & 0.092 & 0.029 & 0.112 & 0.004 \\
0.50 & 10 & 0.070 & 0.173 & 0.066 & 0.226 \\
\bottomrule
\end{tabular}

\end{center}
\vspace{-2mm}
{\small $N=500$ per group. Without self-salt: 7/8 pass. With self-salt: 4/8 pass.}
\end{table}

For the non-self-salt scheme, seven of eight configurations fail to reject the null hypothesis of identical distributions ($p > 0.05$). The single rejection occurs at the same configuration flagged by the t-test in the main text. For the self-salt scheme, four of eight reject equivalence, with the effect stronger at high $\delta$ values where the watermark bias amplifies differences in green list composition due to the hash function (Section~\ref{sec:hash}).

\section{Reproducibility}
\label{app:reproducibility}

\textbf{Hardware.} All experiments were conducted on a single NVIDIA GeForce RTX 3090 GPU (24 GB VRAM, 10,496 CUDA cores, 936 GB/s memory bandwidth).

\textbf{Software.} Python 3.10.12, PyTorch 2.10.0 (CUDA 12.8, cuDNN 9.10.02), NVIDIA driver 535.183.01, vLLM 0.19.1.

\textbf{Generation parameters.} For each prompt, we generated one watermarked and one non-watermarked completion of up to 200 tokens. The same random seed (42) and hash key (15485863) were used across all experiments.

\textbf{Code.} This paper's source and supplementary materials are available at \url{https://github.com/si-mon-jinn/flip-dont-shuffle}, including scripts to reproduce all tables and figures. The watermarking method is implemented in the \texttt{sbw} library (\url{https://github.com/si-mon-jinn/sbw}, also available via \texttt{pip install sbw}). All experiments were run using \texttt{waterpipe} (\url{https://github.com/si-mon-jinn/waterpipe}), a watermark evaluation library.

\section{Isolated Benchmark: Fairness and Methodology}
\label{app:fairness}

This section documents the methodology and fairness considerations of the \emph{isolated logits processor benchmark} (end of Section~\ref{sec:scaling}), which measures per-call watermark kernel latency in isolation from the LLM forward pass. We detail the optimizations applied to each baseline to ensure the comparison reflects each method's best achievable performance.

\subsection{Improvements to the KGW Baseline}

The KGW baseline uses our re-implementation (\texttt{sbw-watermark}, scheme \texttt{ff-additive\_prf-4-False}) that retains the same algorithmic structure as KGW~\cite{kirchenbauer2023watermark,kirchenbauer2024reliability}: sequential \texttt{torch.randperm(V)} per batch element, per-sequence RNG seeding, and Python-level mask construction. We applied several backward-compatible fixes (e.g.\ batched seed computation, avoiding redundant tensor allocations), with the most significant being: the reference selfhash loop uses Python \texttt{enumerate()} over a CUDA tensor, causing one implicit CPU transfer per element. Replacing it with index-based access yields a 23$\times$ speedup while producing bit-identical green lists.

\subsection{Improvements to the SynthID Baseline}

We benchmark against SynthID-Text~\cite{dathathri2024scalable}, vendored from the official \texttt{google-deepmind/synthid-text} repository (Apache 2.0 license). To measure SynthID's \emph{theoretical minimum} latency, we inline all operations into a single \texttt{torch.compile} call with \texttt{mode="max-autotune"}, stripping state management, context history tracking, and safety checks. We verified bit-identical outputs between our compiled pipeline and the reference implementation; tests are available in the paper repository. This gives SynthID every advantage: our reported numbers represent a lower bound on its actual production cost.

\subsection{Our GPU Implementation}

Our fused Bernoulli implementation is benchmarked as-is from the production library (\texttt{sbw-watermark}), with no benchmark-specific optimizations. The compiled kernels are the same ones used in the vLLM logits processor. This means our reported numbers reflect \emph{actual production performance}.

\subsection{Measurement Methodology}

All measurements use identical methodology:
\begin{itemize}
\item CUDA event timing with \texttt{torch.cuda.synchronize()} between iterations
\item 20 warmup iterations (discarded) followed by 200 measurement iterations
\item \texttt{torch.cuda.empty\_cache()} between methods to prevent memory interference
\item Same \texttt{torch.compile} with \texttt{mode="max-autotune"} for all compiled paths
\item Sequential execution (no parallel streams or concurrent kernels)
\item Fresh random logits tensor of shape $(B, 151{,}936)$ for each method
\end{itemize}

\subsection{Known Asymmetries}

\begin{enumerate}
\item \textbf{Depth parameter.} SynthID uses depth$=30$ (Google's official demo configuration, upper end of the recommended 20--30 range). Using depth$=20$ would reduce SynthID's latency by $\sim$33\%.

\item \textbf{Hand-optimized SynthID vs.\ production SBW.} As noted above, our SynthID benchmark represents a theoretical optimum, while our SBW numbers are production code. 
\end{enumerate}


\subsection{Scheme Label Mapping}

Table~\ref{tab:scheme_mapping} maps the paper labels to the library scheme identifiers used in the \texttt{sbw-watermark} package, for reproducibility.

\begin{table*}[h]
\caption{Paper labels and corresponding library scheme identifiers.}
\label{tab:scheme_mapping}
\begin{center}
\footnotesize
\begin{tabular}{@{}llcccc@{}}
\toprule
Paper & Library scheme & PRF & Width & Self-salt & Candidates \\
\midrule
SBW-1 & \texttt{gpu-fused-simple\_1} & additive & 1 & No & full $V$ \\
SBW-4 & \texttt{gpu-fused-simple\_4} & additive & 4 & No & full $V$ \\
SBW-ss & \texttt{gpu-fused-selfhash} & anchored minhash & 4 & Yes & top-40 \\
SBW-ss-V & \texttt{gpu-fused-selfhash-fullvocab} & anchored minhash & 4 & Yes & full $V$ \\
SBW-ss-cpu & \texttt{gpu-fused-selfhash-cpuhash} & anchored minhash & 4 & Yes & top-40 \\
\bottomrule
\end{tabular}
\end{center}
\end{table*}

\section{Isolated Logits Processor Benchmark: Full Results}
\label{app:performance}

This appendix provides the full latency, tail latency, and memory data from the isolated logits processor benchmark (Section~\ref{sec:scaling}), which measures per-call watermark kernel cost on synthetic tensors without an LLM forward pass.

\subsection{Full Latency (All Batch Sizes)}

Table~\ref{tab:full_latency} reports median latency (p50) for all methods across all tested batch sizes. Fig.~\ref{fig:synthid_latency} visualizes the scaling behavior. Fig.~\ref{fig:synthid_speedup} shows speedup ratios on a log scale.

\begin{figure}[!t]
\centerline{\includegraphics[width=\columnwidth]{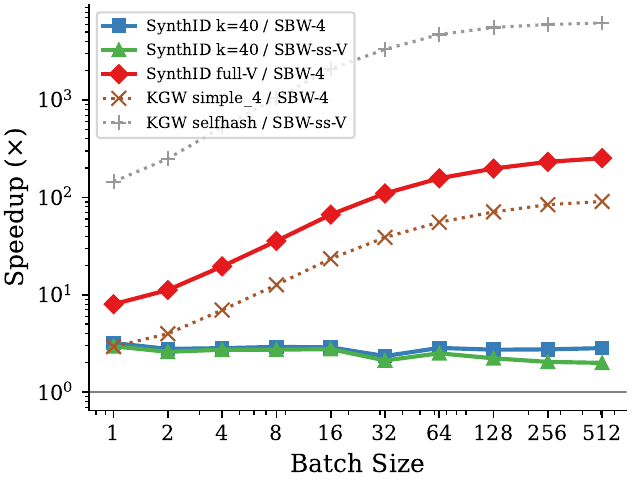}}
\caption{Speedup ratio (log scale). Values above 1.0 indicate slower than the SBW reference.}
\label{fig:synthid_speedup}
\end{figure}

\begin{figure}[!t]
\centerline{\includegraphics[width=\columnwidth]{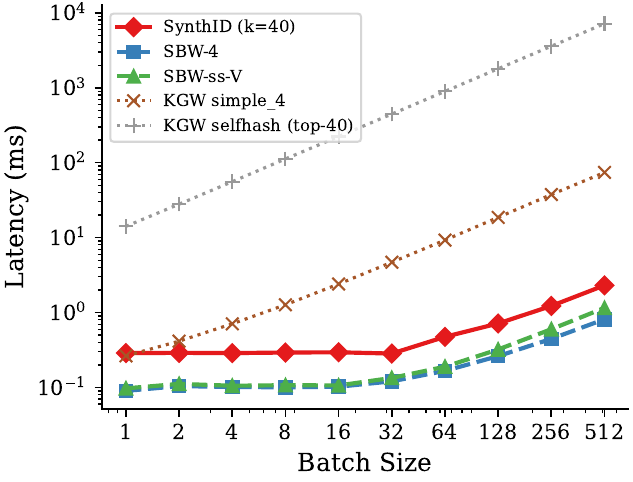}}
\caption{Latency scaling across all methods. SBW-4 and SBW-ss-V (dashed) process the entire 151K vocabulary yet remain faster than SynthID (top-40) at all batch sizes.}
\label{fig:synthid_latency}
\end{figure}

\begin{table*}[!t]
\caption{Median Latency (ms) Across All Batch Sizes}
\label{tab:full_latency}
\begin{center}
\footnotesize
\resizebox{\textwidth}{!}{\begin{tabular}{@{}ll|rrrrrrrrrr@{}}
\toprule
Method & & 1 & 2 & 4 & 8 & 16 & 32 & 64 & 128 & 256 & 512 \\
\midrule
\multirow{2}{*}{SBW-4} & ms & 0.090 & 0.104 & 0.102 & 0.100 & 0.102 & 0.121 & 0.167 & 0.262 & 0.445 & 0.814 \\
 & ratio & 0.9$\times$ & 0.9$\times$ & 1.0$\times$ & 0.9$\times$ & 1.0$\times$ & 0.9$\times$ & 0.9$\times$ & 0.8$\times$ & 0.7$\times$ & 0.7$\times$ \\
\midrule
\multirow{2}{*}{SBW-ss-V} & ms & 0.097 & 0.112 & 0.105 & 0.107 & 0.106 & 0.134 & 0.189 & 0.322 & 0.599 & 1.154 \\
 & ratio & 1.0$\times$ & 1.0$\times$ & 1.0$\times$ & 1.0$\times$ & 1.0$\times$ & 1.0$\times$ & 1.0$\times$ & 1.0$\times$ & 1.0$\times$ & 1.0$\times$ \\
\midrule
\multirow{2}{*}{SynthID (k=40)} & ms & 0.288 & 0.290 & 0.288 & 0.293 & 0.295 & 0.284 & 0.474 & 0.715 & 1.225 & 2.305 \\
 & ratio & 3.0$\times$ & 2.6$\times$ & 2.7$\times$ & 2.7$\times$ & 2.8$\times$ & 2.1$\times$ & 2.5$\times$ & 2.2$\times$ & 2.0$\times$ & 2.0$\times$ \\
\midrule
\multirow{2}{*}{KGW simple\_4} & ms & 0.264 & 0.415 & 0.708 & 1.269 & 2.401 & 4.690 & 9.259 & 18.6 & 37.5 & 73.9 \\
 & ratio & 2.7$\times$ & 3.7$\times$ & 6.7$\times$ & 11.8$\times$ & 22.5$\times$ & 35.0$\times$ & 48.9$\times$ & 57.9$\times$ & 62.7$\times$ & 64.1$\times$ \\
\midrule
\multirow{2}{*}{KGW selfhash (top-40)} & ms & 14.1 & 27.9 & 55.8 & 112.2 & 223.2 & 444.9 & 893.7 & 1790.6 & 3577.5 & 7150.7 \\
 & ratio & 145.2$\times$ & 250.2$\times$ & 528.7$\times$ & 1045.5$\times$ & 2095.3$\times$ & 3318.0$\times$ & 4718.4$\times$ & 5569.4$\times$ & 5972.5$\times$ & 6196.5$\times$ \\
\bottomrule
\end{tabular}
}
\end{center}
\vspace{-2mm}
{\small RTX 3090, V=151,936. Median over 200 iterations with 20 warmup.}
\end{table*}

\subsection{Tail Latency}

Table~\ref{tab:tail_latency} reports p50, p95, and p99 latencies at all batch sizes. Our SBW fullvocab schemes exhibit minimal variance (p99 within 7--15\% of p50), while SynthID and KGW show higher tail latency. Fig.~\ref{fig:tail_latency} visualizes the p99 latency across batch sizes.

\begin{figure}[!t]
\centerline{\includegraphics[width=\columnwidth]{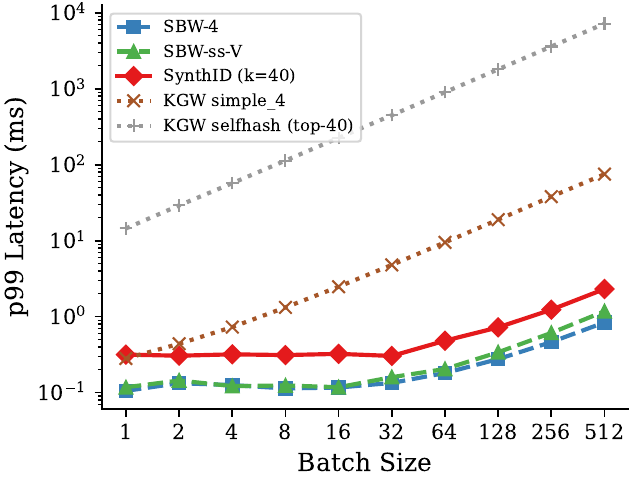}}
\caption{p99 tail latency across batch sizes. SBW-4 and SBW-ss-V maintain low tail latency at all scales.}
\label{fig:tail_latency}
\end{figure}

\begin{table*}[!t]
\caption{Tail Latency (ms): p50 / p95 / p99}
\label{tab:tail_latency}
\begin{center}
\footnotesize
\begin{tabular}{@{}l|rrrrrrrrrr@{}}
\toprule
Method & 1 & 2 & 4 & 8 & 16 & 32 & 64 & 128 & 256 & 512 \\
\midrule
\multicolumn{11}{c}{\textit{p50}} \\
\midrule
SBW-4 & 0.09 & 0.10 & 0.10 & 0.10 & 0.10 & 0.12 & 0.17 & 0.26 & 0.45 & 0.81 \\
SBW-ss-V & 0.10 & 0.11 & 0.11 & 0.11 & 0.11 & 0.13 & 0.19 & 0.32 & 0.60 & 1.15 \\
SynthID (k=40) & 0.29 & 0.29 & 0.29 & 0.29 & 0.29 & 0.28 & 0.47 & 0.71 & 1.22 & 2.31 \\
KGW simple\_4 & 0.26 & 0.42 & 0.71 & 1.27 & 2.40 & 4.69 & 9.26 & 18.6 & 37.5 & 73.9 \\
KGW selfhash (top-40) & 14.1 & 27.9 & 55.8 & 112.2 & 223.2 & 444.9 & 893.7 & 1791 & 3578 & 7151 \\
\midrule
\multicolumn{11}{c}{\textit{p95}} \\
\midrule
SBW-4 & 0.10 & 0.12 & 0.11 & 0.11 & 0.11 & 0.13 & 0.18 & 0.27 & 0.45 & 0.82 \\
SBW-ss-V & 0.11 & 0.13 & 0.12 & 0.12 & 0.12 & 0.14 & 0.20 & 0.33 & 0.61 & 1.17 \\
SynthID (k=40) & 0.30 & 0.30 & 0.30 & 0.31 & 0.31 & 0.30 & 0.48 & 0.72 & 1.23 & 2.31 \\
KGW simple\_4 & 0.28 & 0.43 & 0.72 & 1.31 & 2.46 & 4.78 & 9.45 & 18.8 & 37.8 & 74.7 \\
KGW selfhash (top-40) & 14.3 & 28.5 & 56.5 & 113.0 & 225.0 & 447.4 & 898.4 & 1800 & 3589 & 7184 \\
\midrule
\multicolumn{11}{c}{\textit{p99}} \\
\midrule
SBW-4 & 0.11 & 0.13 & 0.13 & 0.11 & 0.12 & 0.13 & 0.18 & 0.28 & 0.46 & 0.84 \\
SBW-ss-V & 0.12 & 0.14 & 0.12 & 0.12 & 0.12 & 0.16 & 0.20 & 0.34 & 0.61 & 1.18 \\
SynthID (k=40) & 0.32 & 0.31 & 0.32 & 0.31 & 0.32 & 0.30 & 0.48 & 0.72 & 1.24 & 2.31 \\
KGW simple\_4 & 0.28 & 0.44 & 0.73 & 1.32 & 2.48 & 4.80 & 9.50 & 18.9 & 38.0 & 75.0 \\
KGW selfhash (top-40) & 14.5 & 29.1 & 57.1 & 113.3 & 225.7 & 447.9 & 899.5 & 1801 & 3594 & 7191 \\
\bottomrule
\end{tabular}

\end{center}
\vspace{-2mm}
{\small RTX 3090, V=151,936. 200 iterations.}
\end{table*}

\subsection{Memory Overhead}

Table~\ref{tab:memory} reports peak additional GPU memory allocated during watermark computation. Our SBW fullvocab schemes operate in-place on the logits tensor with zero additional allocation. SynthID allocates intermediate tensors proportional to $B \times k \times \text{depth}$.

\begin{table*}[!t]
\caption{Peak Memory Overhead (MB)}
\label{tab:memory}
\begin{center}
\footnotesize
\begin{tabular}{@{}l|rrrrrrr@{}}
\toprule
Method & B=1 & B=8 & B=32 & B=64 & B=128 & B=256 & B=512 \\
\midrule
SBW-4 & 0 & 0 & 0 & 0 & 0 & 0 & 0 \\
SBW-ss-V & 0 & 0 & 0 & 0 & 0 & 0 & 0 \\
\midrule
SynthID (k=40) & 0.08 & 0.27 & 0.32 & 0.43 & 0.86 & 2.1 & 3.4 \\
SBW-4 (k=40) & 0 & 0 & 0 & 0 & 0 & 0 & 0 \\
SBW-ss (k=40) & 0 & 0 & 0 & 0 & 0 & 0 & 0 \\
\bottomrule
\end{tabular}

\end{center}
\vspace{-2mm}
{\small RTX 3090, V=151,936. Delta from baseline (no watermark).}
\end{table*}

\subsection{vLLM Serving Benchmark}
\label{app:vllm}

To validate production deployment overhead, we benchmarked SBW within a vLLM serving stack (Qwen3-8B, 128 requests, 200 output tokens). Table~\ref{tab:vllm_overhead} reports request throughput with and without the SBW logits processor at various concurrency levels.

\begin{table}[h]
\caption{vLLM Serving Overhead}
\label{tab:vllm_overhead}
\begin{center}
\footnotesize
\begin{tabular}{@{}rccc@{}}
\toprule
Concurrency & Baseline (req/s) & SBW (req/s) & Overhead \\
\midrule
1 & 0.23 & 0.23 & $<$0.1\% \\
8 & 1.83 & 1.62 & 11.5\% \\
32 & 6.14 & 6.10 & 0.7\% \\
64 & 9.64 & 9.53 & 1.1\% \\
128 & 13.35 & 13.18 & 1.3\% \\
\bottomrule
\end{tabular}

\end{center}
\vspace{-2mm}
{\small Qwen3-8B, 128 requests, 200 tokens, full-vocabulary self-salt.}
\end{table}

At production concurrency ($\geq$32 concurrent requests), SBW adds ${\sim}$1\% overhead.

\section{Generalization Experiments}
\label{app:generalization}

To validate that SBW's equivalence to KGW generalizes beyond the main experimental setup, we conducted additional experiments varying five factors simultaneously (Table~\ref{tab:generalization_setup}).

\begin{table}[h]
\caption{Generalization Experiment Setup}
\label{tab:generalization_setup}
\begin{center}
\footnotesize
\begin{tabular}{@{}lll@{}}
\toprule
Variable & Main paper & Generalization \\
\midrule
Generation model & Qwen3-8B & Falcon-7B \\
Judge model & Qwen3.5-27B & Mistral-small-3.2-24b \\
Dataset & C4 (web text) & Alpaca (instructions) \\
Sampling & $T{=}1.0$, pure & $T{=}0.7$, $top\_p{=}0.8$ \\
Hardware & RTX 3090 & A6000 \\
\bottomrule
\end{tabular}
\end{center}
\end{table}

We ran 200 samples per scheme at $\delta=2$, $\gamma=0.50$. Table~\ref{tab:generalization} reports two-sample KS tests comparing KGW and SBW z-score distributions. All tests fail to reject equivalence ($p > 0.05$), confirming that the Bernoulli construction's equivalence is independent of model architecture, vocabulary size, prompt distribution, sampling strategy, and hardware.

\begin{table}[h]
\caption{Two-Sample KS Tests: KGW vs SBW on Falcon-7B}
\label{tab:generalization}
\begin{center}
\footnotesize
\begin{tabular}{@{}l|cc|cc@{}}
\toprule
& \multicolumn{2}{c|}{\textit{Watermarked}} & \multicolumn{2}{c}{\textit{Non-watermarked}} \\
Seeding & KS $D$ & $p$ & KS $D$ & $p$ \\
\midrule
Without self-salt & 0.107 & 0.313 & 0.133 & 0.125 \\
With self-salt & 0.114 & 0.662 & 0.080 & 0.676 \\
\bottomrule
\end{tabular}

\end{center}
\vspace{-2mm}
{\small Falcon-7B, Alpaca prompts, $\delta=2$, $\gamma=0.50$, 200 samples.}
\end{table}

\section{Robustness to Attacks}
\label{app:robustness}

To validate that SBW inherits KGW's robustness properties, we evaluated both methods under 19 attack configurations across four categories. \emph{Character-level attacks} randomly substitute or delete individual characters at various fractions (1--20\%). \emph{Structural attacks} truncate the text, keeping only the first 25--75\% of tokens. \emph{Word-level attacks} delete random words (10--30\%) or reorder all words within each sentence. \emph{Semantic attacks} include LLM-based paraphrasing at three intensity levels (Qwen3-4B) and masked language model substitution replacing 10--20\% of words with contextually plausible alternatives (T5). We used 500 watermarked samples per method ($\delta=2$, $\gamma=0.50$, Qwen3-8B).

Table~\ref{tab:robustness} reports the mean z-score after each attack and two-sample KS test p-values comparing the z-score distributions. In 11/19 configurations, the distributions are statistically equivalent ($p > 0.05$). In all 8 non-equivalent cases, SBW retains \emph{more} watermark signal than KGW (higher mean z-score), never less. This is consistent with the hash function analysis in Section~\ref{sec:hash}: the Jenkins hash produces less autocorrelated green-list assignments, so attacks that disrupt local token sequences destroy fewer correlated ``runs'' of green tokens.

\begin{table}[h]
\caption{Robustness: KGW vs SBW under Attacks}
\label{tab:robustness}
\begin{center}
\footnotesize
\resizebox{\columnwidth}{!}{\begin{tabular}{@{}llcccc@{}}
\toprule
Category & Attack & KGW $\bar{z}$ & SBW $\bar{z}$ & KS $p$ & Equiv. \\
\midrule
Char & Subst.\ 1\% & 5.28 & 5.37 & .613 & \checkmark \\
 & Subst.\ 5\% & 3.67 & 3.75 & .413 & \checkmark \\
 & Subst.\ 10\% & 2.40 & 2.36 & .250 & \checkmark \\
 & Subst.\ 20\% & 1.23 & 1.21 & .589 & \checkmark \\
 & Delete 5\% & 3.46 & 3.56 & .534 & \checkmark \\
 & Delete 10\% & 1.96 & 2.11 & .182 & \checkmark \\
 & Delete 20\% & 0.79 & 0.88 & .255 & \checkmark \\
\midrule
Struct & Trunc.\ 25\% & 4.83 & 5.02 & .026 & + \\
 & Trunc.\ 50\% & 4.14 & 4.31 & .107 & \checkmark \\
 & Trunc.\ 75\% & 3.09 & 3.25 & .032 & + \\
\midrule
Word & Delete 10\% & 4.65 & 4.83 & .096 & \checkmark \\
 & Delete 20\% & 4.07 & 4.23 & .050 & + \\
 & Delete 30\% & 3.44 & 3.52 & .559 & \checkmark \\
 & Reorder & 0.70 & 0.96 & {<}.001 & + \\
\midrule
Semantic & Paraph.\ (L) & 1.05 & 1.59 & {<}.001 & + \\
 & Paraph.\ (M) & 1.33 & 1.39 & .842 & \checkmark \\
 & Paraph.\ (H) & 0.72 & 1.37 & {<}.001 & + \\
 & MLM 10\% & 4.48 & 4.71 & .014 & + \\
 & MLM 20\% & 3.81 & 4.01 & .013 & + \\
\bottomrule
\end{tabular}
}
\end{center}
\vspace{-2mm}
{\small 500 samples, $\delta=2$, $\gamma=0.50$. KGW/SBW: mean $z$-score. +: SBW retains more signal.}
\end{table}

\section{Proof of Detection Equivalence}
\label{app:derivation}

\subsection{Setup}

In KGW, the green list $G_t$ at each generation step contains exactly $\gamma|V|$ tokens, selected by a pseudorandom permutation seeded by the context. SBW instead includes each token $i \in V$ independently with probability $\gamma$, producing a stochastic green list proportion:
\begin{equation}
\Gamma_t = \frac{1}{|V|}\sum_{i=1}^{|V|} X_i, \quad X_i \sim \text{Bernoulli}(\gamma) \text{ i.i.d.}
\end{equation}
with $\mathrm{E}[\Gamma_t] = \gamma$ and $\mathrm{Var}(\Gamma_t) = \gamma(1-\gamma)/|V|$.

\subsection{Concentration of Green List Size}

By the Central Limit Theorem, for large $|V|$:
\begin{equation}
\Gamma_t \xrightarrow{d} \mathcal{N}\!\left(\gamma,\; \frac{\gamma(1-\gamma)}{|V|}\right).
\end{equation}
For $\gamma = 0.5$ and $|V| = 151{,}936$, the standard deviation of the green list proportion is $\sigma_\Gamma \approx 1.28$\textperthousand, corresponding to ${\approx}195$ tokens out of ${\approx}75{,}968$.

Hoeffding's inequality provides a finite-sample bound:
\begin{equation}
P(|\Gamma_t - \gamma| > \epsilon) \leq 2\exp(-2|V|\epsilon^2).
\end{equation}
For $\epsilon = 0.01$ (a 1\% deviation): $P(|\Gamma_t - \gamma| > 0.01) \leq 2\exp(-2 \cdot 151{,}936 \cdot 10^{-4}) < 2 \times 10^{-13}$.

\subsection{Symmetry Argument}

Let $Y_t = 1$ if the generated token $s_t$ belongs to $G_t$, and $Y_t = 0$ otherwise. Under $H_0$, the model selects $s_t$ independently of $G_t$. We claim $P(Y_t = 1 \mid \Gamma_t) = \Gamma_t$.

\begin{proof}
By the i.i.d.\ Bernoulli construction, conditioned on $|G_t| = m$, each token $v$ has equal probability $m/|V|$ of being in $G_t$, regardless of which specific tokens were selected. Therefore:
\begin{equation}
P(Y_t = 1 \mid |G_t| = m) = \sum_{v \in V} p_v \cdot \frac{m}{|V|} = \frac{m}{|V|},
\end{equation}
where $p_v$ is the model's probability of generating token $v$. Since $\Gamma_t = m/|V|$, we have $P(Y_t = 1 \mid \Gamma_t) = \Gamma_t$.
\end{proof}

By the Law of Total Expectation, the mean follows immediately: $\mathrm{E}[Y_t] = \mathrm{E}[\mathrm{E}[Y_t \mid \Gamma_t]] = \mathrm{E}[\Gamma_t] = \gamma$.

\subsection{Law of Total Variance}

We decompose $\mathrm{Var}(Y_t)$ by conditioning on $\Gamma_t$:
\begin{equation}
\mathrm{Var}(Y_t) = E[\mathrm{Var}(Y_t \mid \Gamma_t)] + \mathrm{Var}(E[Y_t \mid \Gamma_t]).
\end{equation}

\textbf{First term (conditional variance).} Given $\Gamma_t$, the indicator $Y_t$ is Bernoulli with parameter $\Gamma_t$, so $\mathrm{Var}(Y_t \mid \Gamma_t) = \Gamma_t(1 - \Gamma_t)$. Taking expectations:
\begin{equation}
E[\mathrm{Var}(Y_t \mid \Gamma_t)] = E[\Gamma_t - \Gamma_t^2] = \gamma - E[\Gamma_t^2].
\end{equation}

\textbf{Second term (variance of the conditional mean).} Since $E[Y_t \mid \Gamma_t] = \Gamma_t$:
\begin{equation}
\mathrm{Var}(E[Y_t \mid \Gamma_t]) = \mathrm{Var}(\Gamma_t).
\end{equation}

\textbf{Combining.} Using $E[\Gamma_t^2] = \mathrm{Var}(\Gamma_t) + \gamma^2$:
\begin{align}
\mathrm{Var}(Y_t) &= \gamma - \mathrm{Var}(\Gamma_t) - \gamma^2 + \mathrm{Var}(\Gamma_t) \nonumber\\
                   &= \gamma(1 - \gamma).
\end{align}

The $\mathrm{Var}(\Gamma_t)$ terms cancel exactly. This holds for \emph{any} distribution of $\Gamma_t$ with mean $\gamma$, not only the Bernoulli case. The marginal distribution of $Y_t$ is therefore $\text{Bernoulli}(\gamma)$, identical to the fixed-size KGW green list.

\subsection{Independence of the $Y_t$'s}

The z-score requires that the $Y_t$'s be mutually independent. Under $H_0$:
\begin{itemize}
\item Each $Y_t$ depends only on $G_t$ (determined by a pseudorandom seed $r_t$) and on the model's token choice $s_t$.
\item The model's choice $s_t$ is independent of all green lists by assumption ($H_0$).
\item Distinct seeds $r_t \neq r_{t'}$ produce independent RNG outputs, so $G_t \perp G_{t'}$.
\end{itemize}
Therefore $Y_t \perp Y_{t'}$ for $t \neq t'$, and $W = \sum_{t=1}^T Y_t \sim \text{Binomial}(T, \gamma)$.

\end{document}